\documentclass[11pt]{amsart}
\usepackage{graphicx} 
\usepackage[labelfont=bf]{caption}
\usepackage{amsfonts}
\usepackage{amsmath}
\usepackage{amssymb}
\usepackage{amsthm}
\usepackage{subcaption}
\usepackage{booktabs}
\usepackage{float}
\usepackage{enumitem}
\usepackage{color,soul}
\usepackage[labelfont=bf]{caption}
\usepackage{appendix}
\usepackage{float}
\usepackage{setspace}
\usepackage{xcolor}
\usepackage{lineno}
\usepackage{hyperref}
\usepackage{todonotes}

\newcommand{\R}{\mathbb{R}}
\newcommand{\F}{\mathcal{F}}
\newcommand{\Exp}{\mathbb{E}}
\newcommand{\pp}{\mathbb{P}}
\newcommand{\Q}{\mathbb{Q}}

\newtheorem{theorem}{Theorem}[section]

\newtheorem{proposition}[theorem]{Proposition}
\newtheorem{remark}[theorem]{Remark}

\providecommand{\keywords}[1]
{
  \small
  \textbf{\textbf{Keywords---}} #1
}

\allowdisplaybreaks

\title{Call Option Price using Pearson Diffusion Processes}

\author{Tapan Kar}
\address{Tapan Kar, Institute for Financial Management and Research - Graduate School of Business, Krea University, Sri City, Andhra Pradesh 524401, India.}
\email{tapan.kar@krea.edu.in}

\author{Suprio Bhar}
\address{Suprio Bhar, Department of Mathematics and Statistics, Indian Institute of Technology Kanpur, Kalyanpur, Kanpur - 208016, India.}
\email{suprio@iitk.ac.in}

\author{Barun Sarkar}
\address{Barun Sarkar, Department of Mathematics, Indian Institute of Technology Madras, Chennai - 600036, India.}
\email{barun@iitm.ac.in}

\author{Sesha Meka}
\address{Sesha Meka, Institute for Financial Management and Research - Graduate School of Business, Krea University, Sri City, Andhra Pradesh 524401, India.}
\email{sesha.meka@krea.edu.in}

\date{}

\begin{document}

\begin{abstract}
Following the foundational work of the Black--Scholes model, extensive research has been developed to price the option by addressing its underlying assumptions and associated pricing biases. This study introduces a novel framework for pricing European call options by modeling the underlying asset's return dynamics using Pearson diffusion processes, characterized by a linear drift and a quadratic squared diffusion coefficient. This class of diffusion processes offers a key advantage in its ability to capture the skewness and excess kurtosis of the return distribution, well-documented empirical features of financial returns. We also establish the validity of the risk-neutral measure by verifying the Novikov condition, thereby ensuring that the model does not admit arbitrage opportunities. Further, we study the existence of a unique strong solution of stock prices under the risk-neutral measure. We apply the proposed method to Nifty 50 index option data and conduct a comparative analysis against the classical Black--Scholes and Heston stochastic volatility models. Results indicate that our method shows superior performance compared to the other two benchmark models. These results carry practical implications for market participants, including market makers, hedge funds, and derivative traders.
\end{abstract}



\maketitle

\section{Introduction}

Black and Scholes \cite{MR3363443} and Merton \cite{MR496534} developed option pricing models based on the assumption that the underlying asset price follows a geometric Brownian motion, implying that its log returns are normally distributed. However, empirical evidence shows that log returns deviate from normality, displaying skewness and excess kurtosis. Consequently, the option prices estimated by these models often deviate from observed market prices.

Researchers have proposed various alternative option pricing models to address the limitations of the Black-Scholes framework. An alternative to the Black-Scholes model is the constant elasticity of variance (CEV) approach, which generalizes geometric Brownian motion by allowing the diffusion coefficient to follow a power-law relationship with the stock price. This approach has been adopted by several researchers, including Cox \cite{cox1975notes,cox1996constant}, Cox and Ross \cite{cox1976valuation}, and Beckers \cite{beckers1980constant}, among others. Merton \cite{merton1976option} extended the Black-Scholes model by introducing a jump-diffusion framework incorporating a jump component into the asset price dynamics.
Heston \cite{MR3929676}  introduced the stochastic volatility model in which volatility itself is governed by a separate stochastic process. Madan \textit{et al.} \cite{madan1998variance} proposed the Variance Gamma option pricing model, driven by a stochastic process known as the Variance Gamma process, which is constructed by evaluating Brownian motion with drift at a random time given by a Gamma process. Cervellera and Tucci \cite{cervellera2017note} discussed challenges associated with parameter estimation for the Variance Gamma process. Borland \cite{borland2002theory} proposed an option pricing model based on a diffusion process described by a stochastic differential equation whose noise was driven by a non-Gaussian distribution, and the resulting distribution of log return process over time can be explicitly characterized by the Tsallis distribution \cite{tsallis1996anomalous}. However, Vellekoop and Nieuwenhuis \cite{vellekoop2007option} later argued that Borland’s model does not satisfy the Novikov condition, a key requirement for risk-neutral pricing, thereby allowing the possibility of arbitrage.
Additionally, Karagozoglu \cite{karagozoglu2022option} provided a good literature review on option pricing models.

In option pricing, researchers aim to model log returns using distributions that can capture the excess kurtosis and skewness of the return process. The Pearson Type IV distribution is well-suited for this task, as it accommodates a broad range of kurtosis and skewness depending on its shape parameter values. Nagahara \cite{MR1708093} provided a comprehensive analysis of this distribution, deriving its normalizing constant, the first four moments, and illustrating its applicability to modeling stock return distributions. Bhattacharyya \textit{et al.} \cite{bhattacharyya2009maxvar} employed the Pearson Type IV distribution to forecast Value at Risk for leptokurtic equity index returns. Further, Forman and Sørensen \cite{MR2446729} investigated the statistical properties of Pearson diffusion processes, a class of diffusion characterized by linear drift and a quadratic squared diffusion coefficient and showed that under specific conditions on the diffusion coefficient, the process follows Pearson Type IV distribution as its invariant distribution.

The main contribution of this paper is introducing a novel method for estimating call option prices by modeling log returns using Pearson diffusion processes, henceforth we refer to this method as the PIV model. In addition, we establish the validity of the Novikov condition within this framework, crucial for risk-neutral pricing, and study the existence of a unique strong solution for the stock price process under the risk-neutral measure. We apply the proposed method to Nifty 50 index option data and perform a comparative analysis against the classical Black–Scholes and Heston stochastic volatility models. The results demonstrate that our approach outperforms both benchmark models.

The rest of the paper is organized in the following way. Section \ref{pearson} describes the Pearson Type IV distribution. Section \ref{pearson-difSec} presents the Pearson diffusion processes. Section \ref{returnmodelP} outlines the modeling of log return processes. In Section \ref{equivMartMeas}, we derive the risk-neutral measure and present the proposed option pricing model. Section \ref{emperical} provides empirical evidence, and Section \ref{conclusion} concludes the paper.

\section{Pearson Type IV distribution}\label{pearson}

Karl Pearson \cite{pearson1895contributions} was the first to develop a generalized family of frequency curves designed to accommodate a wide range of empirical distributions encountered in practice. This comprehensive framework is now recognized as the Pearson system of curves. The Pearson system of curves satisfy the differential equation defined as:

\begin{equation}\label{ode1}
\frac{1}{f(x)}\, \frac{df(x)}{dx} = \frac{x - \alpha}{c_0 + c_1 x + c_2 x^2}
\end{equation}

A general solution of the above differential equation is given by
\begin{equation}\label{ode2}
f(x) = k \exp \left( \int \frac{x - \alpha}{c_0 + c_1 x + c_2 x^2} \, dx \right),
\end{equation}
where \( k \) is a constant.

The explicit form of the function  $f(x)$ is governed by the integral in the exponent of \eqref{ode2}. The evaluation of this integral, in turn, depends on the nature of the roots of the equation $c_0 + c_1 x + c_2 x^2 = 0$. The Pearson Type IV curve is obtained when the roots are imaginary, or when $c_1^2 < 4 c_0 c_2$.

The probability density function (PDF) of Pearson’s Type IV distribution is given by
\begin{equation}\label{pdf-pearson4}
f(x)= k \left[ 1 + \left( \frac{x - \lambda}{a} \right)^2 \right]^{-m} \exp\left[ \nu \tan^{-1}\left( \frac{x - \lambda}{a} \right) \right]  \quad -\infty < x < \infty,
\end{equation}
where $m > \frac{1}{2}$, $a>0$, and $m$, $a$, $\nu$ and $\lambda$ are real-valued parameters (functions of $\alpha$, $c_0$, $c_1$, and $c_2$) and $k$ is a normalization constant that depends on $m$, $a$, and $\nu$. Nagahara \cite{MR1708093} derived a closed-form expression of the normalization constant $k$.

The parameters $\nu$ and $m$ control skewness and kurtosis, respectively. Accordingly, the distribution is positively (negatively) skewed as $\nu$ is negative (positive). Increasing $m$ reduces the kurtosis. The scale parameter is $a$, and the location parameter is $\lambda$.  The Pearson Type IV distribution can capture a wide range of shapes depending on its parameter values.

\section{Pearson Diffusion Processes}\label{pearson-difSec}

Forman and Sørensen \cite{MR2446729} studied the statistical properties of the Pearson diffusion processes, which are defined by having a linear drift and a quadratic squared diffusion coefficient.

Let $\pp$ be a probability measure on a sample space $\Omega$. Pearson diffusion is a stationary solution to the following stochastic differential equation (SDE):
\begin{equation}\label{pearson-difforg}
dR_t = -\theta(R_t-\mu)\, dt +  \sqrt{2\theta (aR_t^2+bR_t+c)}\, dB_t
\end{equation}
where $\theta > 0$ and $a$, $b$, and $c$ are such that the square root term is well defined. Here $B_{t}, t \geq 0$ is a Brownian motion and $\mu$ denotes the mean of the invariant distribution.

As a particular case, when $a > 0$, $b = 0$, and $a = c$, the SDE \eqref{pearson-difforg} simplifies to the form:
\begin{equation}\label{pearson-diff2}
dR_t = -\theta(R_t-\mu)\, dt +  \sqrt{2\theta a(1+R_t^2)}\, dB_t
\end{equation}

Forman and Sørensen \cite{MR2446729} stated that when $\mu = 0$, the invariant distribution of the SDE \eqref{pearson-diff2} is a scaled  $t$-distribution with $m_{1} = 1 + a^{-1}$  degrees of freedom and a scale parameter $m_{1}^{-1/2}$. In the case where $\mu \neq 0$, the invariant distribution corresponds to the Pearson Type IV distribution, with $\mu$ controlling the skewness and the tails decaying at the same rate as those of a $t$-distribution with $1 + a^{-1}$ degrees of freedom. In both cases, the invariant distribution has moments of order $k_1$ for $k_1 < 1 + a^{-1}$.


\section{Model of returns: Existence and uniqueness of  solution for the log return process}\label{returnmodelP}

Let the stock price dynamics be given by,
\[S_t:= S_0e^{R_t}, \, \forall t \geq 0,\]
with $S_0$ being the initial stock price and $R_t, t \geq 0$ be the $\log$ return process. Note that $R_t=\ln \frac{S_t}{S_0}$ and  $R_0=0$.

Let $R_t, t \geq 0$ follow Pearson diffusion process:
\begin{equation}\label{pearson-diff1}
dR_t = -\theta(R_t-\mu)\, dt + \sigma \sqrt{2\theta a(1+R_t^2)}\, dB_t,\, t \geq 0
\end{equation}
where $\sigma$ denotes the volatility of $\log$ returns. We use the parameter $\sigma$ for scaling the invariant distribution.

Next, we study the existence of strong global solution of \eqref{pearson-diff1}. For that we have to show the drift and diffusion coefficients are globally Lipschitz  \cite[Theorem 5.2.1]{MR2001996}. The drift and diffusion coefficients are given by
\[u(x):= -\theta(x-\mu),\ \ v(x):=\sigma \sqrt{2\theta a(1+x^2)},\, \forall x \in \R,\]
respectively. Note that,
\[|u(x) -u(y)| = \theta |x-y|, \, \forall x, y \in \R.\]
Again,
\begin{align*}
|v(x) -v(y)| & = \sigma\sqrt{2\theta a} \left| \sqrt{1+x^2} - \sqrt{1+y^2}\right| \\
& = \sigma\sqrt{2\theta a} \frac{|x^2-y^2|}{\left| \sqrt{1+x^2} + \sqrt{1+y^2}\right|} \\
& \leq \frac{\sigma\sqrt{2\theta a}\, |x+y|}{ \sqrt{1+x^2} + \sqrt{1+y^2}}\, |x-y| \\
& \leq \sigma\sqrt{2\theta a}\, |x-y|, \, \forall x, y \in \R.
\end{align*}
The last inequality follows from the estimate
\[|x+y|\leq |x|+|y|\leq   \sqrt{1+x^2} + \sqrt{1+y^2} .\]
Therefore there exists an unique strong solution of \eqref{pearson-diff1}. In particular by \cite[Theorem 5.2.1]{MR2001996},
\begin{equation}\label{Rt-integrability}
\pp\left( \int_0^T R_t^2\, dt<\infty\right) =1.
\end{equation}
We now derive the SDE governing the stock price using the It\^o's formula. Consider the function $f(x)=e^x, x \in \R$. Then, we have, a.s.,
\begin{align*}
dS_t & = S_0 d\left( e^{R_t}\right) \\
& = S_t \left\{ -\theta(R_t-\mu) +  \sigma^2\theta a \left(1+R_t^2 \right)\right\}\, dt + S_t \sigma \sqrt{2\theta a(1+R_t^2)}\, dB_t.
\end{align*}

\section{Equivalent local martingale measure}\label{equivMartMeas}

Consider a European-style call option with strike price $K$, time to maturity $T$, written on an underlying asset $S_t$, $0\leq t\leq T$. The payoff from the European call option at maturity is $\max(S_{T}-K,0)$. At a time $t < T$, using the information $I_{t}$ available up to time $t$ (more precisely, a filtration $\{I_t\}_{t\leq T}$) the fair market price of the call option denoted by $C(t)$ is equal to the discounted (under the risk-free interest rate $r$) expected payoff, and is given as:
\begin{equation}\label{defination_call}
C(t) = e^{-r(T-t)}\mathbb{E}^\mathbb{Q}[\max(S_T - K, 0)\mid I_t],
\end{equation}
where $\mathbb{E}^\mathbb{Q}$  denotes expectation under the risk-neutral probability measure $\mathbb{Q}$. The market completeness ensures the uniqueness of the risk-neutral measure (see \cite{constantinides2007option}).
The equality stated in \eqref{defination_call} holds provided that the discounted asset price process $e^{-rt} S_t,\ t\leq T$ is a martingale under $\mathbb{Q}$.  In this case, $\mathbb{Q}$ is called an equivalent martingale measure. The probability measure $\mathbb{Q}$ is called an equivalent local martingale measure if $e^{-rt}S_{t},\ t\leq T$ becomes a local martingale under $\mathbb{Q}$ (see \cite[Remark, page 165]{MR2001996}).

According to the asset pricing theory, to compute the expectation under the risk-neutral measure, we first convert $e^{-rt} S_t,\ t\leq T$ into a local martingale under the risk-neutral probability measure $\Q$. To identify $\mathbb{Q}$, we use the celebrated Girsanov's Theorem.

\begin{theorem}[Girsanov's Theorem]\cite[Theorem 8.6.6]{MR2001996}\label{Girsanov-theorem}
Let $Y_t, t \leq T$ be a real valued It\^o process of the form
\[dY_t = \beta(t)\, dt + \theta(t) \, dB_t;\, t \leq T.\]
Suppose there exist processes $u_t, t \leq T$ and $\alpha_t, t \leq T$ such that
\[\theta_t u_t = \beta_t - \alpha_t\]
and assume that $u_t, t\leq T$ satisfies the Novikov's condition
\[\Exp \left[\exp\left( \frac{1}{2} \int_0^T u^2_s\, ds   \right)\right] < \infty.\]
Put
\[M_t := \exp\left(- \int_0^t u_s \, dB_s - \frac{1}{2} \int_0^t u^2_s\, ds \right),\, t \leq T\]
and
$d \Q(\omega) = M_T (\omega) \, d\pp(\omega)$
on $\F_T$. Then,
\[\widehat B_t := \int_0^t u_s \, ds + B_t,\, t \leq T\]
is a Brownian motion with respect to $\Q$ and in terms of $\widehat B_t, t \leq T$, the process $Y_t, t \leq T$ has the stochastic integral representation
\[dY_t = \alpha(t)\, dt + \theta(t) \, d\widehat B_t;\, t \leq T.\]
\end{theorem}

\begin{remark}
In Theorem \ref{Girsanov-theorem}, we require $M_t, t \leq T$ to be a $\pp$ martingale, specifically an exponential martingale, see \cite[Chapter III, Theorem 5.3]{MR1011252}.
\end{remark}

We now derive the SDE for the discounted price process $e^{-rt}S_{t},\, t\leq T$ under the original probability measure $\pp$. Consider the function $f(t,x):=e^{-rt}x, t \geq 0, x \in \R$. Then,
\[\frac{\partial f}{\partial t}(t,x)=-re^{-rt}x,\ \ \frac{\partial f}{\partial x}(t,x)=e^{-rt},\ \ \frac{\partial^2f}{\partial x^2}(t,x)=0, \, \forall t,x.\]
Therefore,
\begin{align}\label{discounted-price-under-P}
\begin{split}
d\left( e^{-rt} S_t\right) & = -re^{-rt} S_t\, dt + e^{-rt} S_t \left\{ -\theta(R_t-\mu) +  \sigma^2\theta a \left(1+R_t^2 \right)\right\}\, dt  \\
& \quad + e^{-rt} S_t \sigma \sqrt{2\theta a(1+R_t^2)}\, dB_t
\end{split}
\end{align}
Girsanov's Theorem \ref{Girsanov-theorem} proposed a new Brownian motion $\widehat B_t, t \leq T$ under $\Q$ is of the form
\[d\widehat{B_t}= u_t\, dt + dB_t, t \leq T.\]
Therefore, \eqref{discounted-price-under-P} can be written as:
\begin{align}\label{eq:discounted_process_drift_zero}
d\left( e^{-rt} S_t\right) & = -re^{-rt} S_t\, dt + e^{-rt} S_t \left\{ -\theta(R_t-\mu) +  \sigma^2\theta a \left(1+R_t^2 \right)\right\}\, dt  \notag \\
& \quad + e^{-rt} S_t \sigma \sqrt{2\theta a(1+R_t^2)}\, dB_t \notag \\
& = -re^{-rt} S_t\, dt + e^{-rt} S_t \left\{ -\theta(R_t-\mu) +  \sigma^2\theta a \left(1+R_t^2 \right)\right\}\, dt  \notag \\
& \quad + e^{-rt} S_t \sigma \sqrt{2\theta a(1+R_t^2)}\, \left(d\widehat{B_t} - u_t \, dt\right)
\end{align}
To ensure that the discounted price process $e^{-rt} S_t,\ t\leq T$ becomes a local martingale under $\mathbb{Q}$, we set the drift term of \eqref{eq:discounted_process_drift_zero} to zero and obtain:
\begin{equation}\label{expression-for-ut}
    u_t = \frac{-r  -\theta(R_t-\mu) +  \sigma^2\theta a \left(1+R_t^2 \right)}{ \sigma \sqrt{2\theta a(1+R_t^2)}}.
\end{equation}

\begin{theorem}\label{localmart-discountStck}
There exists a equivalent local martingale measure $\Q$ such that such that under $\Q$, the discounted price process $e^{-rt} S_t,\, t \leq T$ becomes a local martingale.
\end{theorem}

\begin{proof} For the proof, see Appendix \ref{App}.
\end{proof}

\subsection{Black-Scholes type PDE}

Let $F:[0, \infty) \times \R \to \R$ be $C^{1, 2}$ regular. Here, $C^{1, 2}$ denotes the space of two variable functions, once continuously differentiable in the first variable and twice continuously differentiable in the second variable. For any derivative price at time $t$ given by $F(t, S_t)$, by It\^o's formula,
\begin{align*}
& dF(t,S_t)\\
& = \frac{\partial F}{\partial t}(t,S_t)\, dt + \frac{\partial F}{\partial S}(t,S_t)\, dS_t + \frac{1}{2}\, \frac{\partial^2 F}{\partial S^2}(t,S_t)\, d[S,S]_t \\
& = \frac{\partial F}{\partial t}(t,S_t)\, dt\\
& \quad + \frac{\partial F}{\partial S}(t,S_t)\, \left\{ S_t \left\{ -\theta(R_t-\mu) +  \sigma^2\theta a \left(1+R_t^2 \right)\right\}\, dt + S_t \sigma \sqrt{2\theta a(1+R_t^2)}\, dB_t\right\} \\
& \quad + \frac{1}{2}\, \frac{\partial^2 F}{\partial S^2}(t,S_t)\, \sigma^2 S_t^2 \left( 2\theta a(1+R_t^2)\right)\, dt\\
& = \left\{ \frac{\partial F}{\partial t}(t,S_t) + \frac{\partial F}{\partial S}(t,S_t)\, \left[ S_t \left\{ -\theta(R_t-\mu) +  \sigma^2\theta a \left(1+R_t^2 \right)\right\}  \right] + \frac{1}{2}\, \frac{\partial^2 F}{\partial S^2}(t,S_t)\, \sigma^2 S_t^2 \left( 2\theta a(1+R_t^2)\right)\right\}\, dt\\
& \quad +  \frac{\partial F}{\partial S}(t,S_t)\, S_t \sigma \sqrt{2\theta a(1+R_t^2)}\, dB_t
\end{align*}
For the discounted derivative price, using $d\widehat{B_t}$ from Theorem \ref{Girsanov-theorem} and $u_t$ from \eqref{expression-for-ut}, we obtain
\begin{align*}
& d \left( e^{-rt} F(t,S_t)\right) \\
& = -re^{-rt}  F(t,S_t)\, dt + e^{-rt} dF(t,S_t) \\
& = e^{-rt} \Big\{ -r F(t,S_t) +  \frac{\partial F}{\partial t}(t,S_t) + \frac{\partial F}{\partial S}(t,S_t)\, \left[ S_t \left\{ -\theta(R_t-\mu) +  \sigma^2\theta a \left(1+R_t^2 \right)\right\}  \right] \\
& \hspace{2cm} + \frac{1}{2}\, \frac{\partial^2 F}{\partial S^2}(t,S_t)\, \sigma^2 S_t^2 \left( 2\theta a(1+R_t^2)\right) \Big\}\, dt \\
& \quad + e^{-rt} \frac{\partial F}{\partial S}(t,S_t)\, S_t \sigma \sqrt{2\theta a(1+R_t^2)}\, dB_t \\
& = e^{-rt} \Big\{ -r F(t,S_t) +  \frac{\partial F}{\partial t}(t,S_t) + \frac{\partial F}{\partial S}(t,S_t)\, \left[ S_t \left\{ -\theta(R_t-\mu) +  \sigma^2\theta a \left(1+R_t^2 \right)\right\}  \right] \\
& \hspace{2cm} + \frac{1}{2}\, \frac{\partial^2 F}{\partial S^2}(t,S_t)\, \sigma^2 S_t^2 \left( 2\theta a(1+R_t^2)\right) \Big\}\, dt \\
& \quad + e^{-rt} \frac{\partial F}{\partial S}(t,S_t)\, S_t \sigma \sqrt{2\theta a(1+R_t^2)}\, d\widehat{B_t} - e^{-rt} \frac{\partial F}{\partial S}(t,S_t)\, S_t \sigma \sqrt{2\theta a(1+R_t^2)}\,  u_t\,dt, \\
& = e^{-rt} \Big\{ -r F(t,S_t) +  \frac{\partial F}{\partial t}(t,S_t) + \frac{\partial F}{\partial S}(t,S_t)\, \left[ S_t \left\{ -\theta(R_t-\mu) +  \sigma^2\theta a \left(1+R_t^2 \right)\right\}  \right] \\
& \hspace{2cm} + \frac{1}{2}\, \frac{\partial^2 F}{\partial S^2}(t,S_t)\, \sigma^2 S_t^2 \left( 2\theta a(1+R_t^2)\right) \Big\}\, dt \\
& \quad + e^{-rt} \frac{\partial F}{\partial S}(t,S_t)\, S_t \sigma \sqrt{2\theta a(1+R_t^2)}\, d\widehat{B_t} \\
& \quad - e^{-rt} \frac{\partial F}{\partial S}(t,S_t)\, S_t \sigma \sqrt{2\theta a(1+R_t^2)}\, \left(\frac{-r  -\theta(R_t-\mu) +  \sigma^2\theta a \left(1+R_t^2 \right)}{ \sigma \sqrt{2\theta a(1+R_t^2)}}\right)\, dt \\
& = e^{-rt} \Big\{ -r F(t,S_t) +  \frac{\partial F}{\partial t}(t,S_t) + rS_t \frac{\partial F}{\partial S}(t,S_t) +  \frac{\partial^2 F}{\partial S^2}(t,S_t)\, \sigma^2 S_t^2 \left( \theta a(1+R_t^2)\right) \Big\}\, dt  \\
&\quad + e^{-rt} \frac{\partial F}{\partial S}(t,S_t)\, S_t \sigma \sqrt{2\theta a(1+R_t^2)}\, d\widehat{B_t}.
\end{align*}
Equating the drift term to be zero and converting it to the deterministic condition, we end up with a Black-Scholes type PDE
\begin{equation}\label{PIVPDE}
    \frac{\partial F}{\partial t}(t,x) + rx \frac{\partial F}{\partial x}(t,x) + \sigma^2 x^2 \left( \theta a\left(1+\left(\ln \frac{x}{S_0}\right)^2 \right)\right)\,  \frac{\partial^2 F}{\partial x^2}(t,x) = r F(t,x).
\end{equation}
The above discussion is formulated in the following proposition.
\begin{proposition}
Assume that the log return follows Pearson diffusion process \eqref{pearson-diff1}. Then, the Black-Scholes type PDE for derivative pricing is \eqref{PIVPDE}.
\end{proposition}

\subsection{Existence and Uniqueness of stock price under \texorpdfstring{$\Q$}{Q}}

Under $\Q$, the log return process satisfies
\begin{align}\label{Rtquadeqn}
 \begin{split}
 dR_t & = -\theta(R_t-\mu)\, dt + \sigma \sqrt{2\theta a(1+R_t^2)}\, dB_t \\
 & = -\theta(R_t-\mu)\, dt + \sigma \sqrt{2\theta a(1+R_t^2)}\, \left(d\widehat{B_t} - u_t\, dt\right) \\
 & = -\theta(R_t-\mu)\, dt + \sigma \sqrt{2\theta a(1+R_t^2)}\, d\widehat{B_t} \\
 & \quad - \sigma \sqrt{2\theta a(1+R_t^2)}\,  \left(\frac{-r  -\theta(R_t-\mu) +  \sigma^2\theta a \left(1+R_t^2 \right)}{ \sigma \sqrt{2\theta a(1+R_t^2)}}\right)\, dt \\
 & = \left( r - \sigma^2\theta a \left(1+R_t^2 \right)\right)\, dt + \sigma \sqrt{2\theta a(1+R_t^2)}\, d\widehat{B_t},
\end{split}
\end{align}
Let $\Pi(x)$ be the probability density function corresponding to the steady state distribution of \eqref{Rtquadeqn} under the probability measure $\mathbb{Q}$. Then $\Pi(x)$ satisfies the  Fokker-Planck equation defined as:
\[ \frac{\partial^2}{\partial x^2} \left\{\frac{1}{2} \sigma^22\theta a\left(1+x^2\right) \Pi(x) \right\} = \frac{\partial}{\partial x} \left\{ \left( r - \sigma^2\theta a \left(1+x^2 \right)\right)\Pi(x) \right\},\]
simplifying which we get,
\begin{equation}\label{inv-dist1}
\left(1+x^2\right) \Pi^{\prime\prime}(x) + \left\{ 4x - \frac{r}{\sigma^2\theta a} + \left(1+x^2 \right) \right\}\Pi^\prime(x) + 2 (1+x)\Pi(x) = 0.
\end{equation}
One can get the steady state distribution of \eqref{Rtquadeqn} by solving the above differential equation \eqref{inv-dist1}. Here, we are unable to obtain any closed form solution of \eqref{inv-dist1}.

Next we discuss the stock price dynamics under $\mathbb{Q}$,
\begin{align}\label{StinQ}
\begin{split}
 & dS_t \\
 & = S_t \left\{ -\theta(R_t-\mu) +  \sigma^2\theta a \left(1+R_t^2 \right)\right\}\, dt + S_t \sigma \sqrt{2\theta a(1+R_t^2)}\, dB_t \\
& = S_t \left\{ -\theta(R_t-\mu) +  \sigma^2\theta a \left(1+R_t^2 \right)\right\}\, dt + S_t \sigma \sqrt{2\theta a(1+R_t^2)}\, \left(d\widehat{B_t} - u_t\, dt\right) \\
& = S_t \left\{ -\theta(R_t-\mu) +  \sigma^2\theta a \left(1+R_t^2 \right)\right\}\, dt + S_t \sigma \sqrt{2\theta a(1+R_t^2)}\, d\widehat{B_t} \\
& \quad - S_t \sigma \sqrt{2\theta a(1+R_t^2)}\, \left(\frac{-r  -\theta(R_t-\mu) +  \sigma^2\theta a \left(1+R_t^2 \right)}{ \sigma \sqrt{2\theta a(1+R_t^2)}}\right)\, dt \\
& = rS_t\, dt + S_t \sigma \sqrt{2\theta a(1+R_t^2)}\, d\widehat{B_t} \\
& = rS_t\, dt + S_t \sigma \sqrt{2\theta a\left(1+ \left(\ln \frac{S_t}{S_0}\right)^2 \right)}\, d\widehat{B_t} .
\end{split}
\end{align}
By Girsanov's Theorem, the SDE \eqref{StinQ} under $\mathbb{Q}$ has a weak solution. In the next theorem, we establish the existence and uniqueness of strong solution of this equation.
\begin{theorem}
The SDE \eqref{StinQ} has a unique strong solution.
\end{theorem}
\begin{proof}
    Consider the diffusion coefficient -
\[\hat\sigma (x) := x\sigma \sqrt{2\theta a\left(1+ \left(\ln \frac{x}{S_0}\right)^2 \right)}, \, \forall x \in \R.\]
Then,
\begin{align*}
\hat\sigma^\prime(x) & =\sigma \sqrt{2\theta a\left(1+ \left(\ln \frac{x}{S_0}\right)^2 \right)} + x\, \sigma\, \frac{1}{2}\,  \frac{1}{\sqrt{2\theta a\left(1+ \left(\ln \frac{x}{S_0}\right)^2 \right)}}\, 2\left(\ln\frac{x}{S_0}\right)\, \frac{1}{x}\\
& = \sigma \sqrt{2\theta a\left(1+ \left(\ln \frac{x}{S_0}\right)^2 \right)} + \frac{\sigma\, \ln\frac{x}{S_0}}{\sqrt{2\theta a\left(1+ \left(\ln \frac{x}{S_0}\right)^2 \right)}},
\end{align*}
and consequently,
\[\left| \hat\sigma^\prime(x)\right| \leq \sigma \sqrt{2\theta a\left(1+ \left(\ln \frac{x}{S_0}\right)^2 \right)} + \frac{\sigma}{\sqrt{2\theta a}}.\]
Now, for any large $L > 0$, with $|x|,|y|\leq L$, by the Mean Value Theorem, we have
\begin{align*}
& \left| \hat\sigma (x) - \hat\sigma (y)\right| \\
& = \left| \hat\sigma^\prime(z)\right|\, |x-y|,\ \ \text{[for $z\in(x,y)\Rightarrow |z|\leq L$]}\\
& \leq \left( \sigma \sqrt{2\theta a\left(1+ \left(\ln \frac{L}{S_0}\right)^2 \right)} + \frac{\sigma}{\sqrt{2\theta a}}\right)\, |x-y|.
\end{align*}
Therefore, the diffusion coefficient is locally Lipschitz and the drift coefficient $\hat b(x) := rx, x \in \R$ is clearly globally Lipschitz. Therefore, from \cite[Chapter 5, Theorem 2.5]{MR1121940}, we have the local existence and uniqueness for strong solution of $S_t$.

Now, we verify condition (3.13) of  \cite{MR3131308}. For $x^2 > e$,
\begin{align*}
& r\hat b(x) - \frac{1}{2} \hat\sigma^2(x) \\
& = rx^2 - \frac{1}{2}\, x^2\sigma^2 2\theta a\left(1+ \left(\ln \frac{x}{S_0}\right)^2 \right) \\
& = (r - \sigma^2\theta a)\, x^2 - x^2\sigma^2 \theta a \left(\ln \frac{x}{S_0}\right)^2  \\
& \leq \left| r - \sigma^2\theta a\right|\, x^2 \\
& \leq \left| r - \sigma^2\theta a\right|\, x^2\, \ln (x^2),
\end{align*}
where we use the fact that $\ln (x^2)>1$. Therefore, we conclude that the solution is always regular i.e. non-explosive. This completes the proof.
\end{proof}

\begin{remark}
By the existence and uniqueness of the strong solution discussed above, we have the global solution of $S_t,\, t \leq T$. However, in our approach we have not been able to establish the integrability of $S_t$ under $\mathbb{Q}$. As a result, $e^{-rt} S_t, t \leq T$ becomes a local martingale under $\Q$ and we have not been able to establish the martingale property.
\end{remark}

\subsection{Estimation of Call option price}
In what follows, we assume that $e^{-rt} S_t, t \leq T$ is a $\Q$ martingale. The price of a call option at time
$t$, with maturity at $T$ ($t <T$), strike price $K$, and risk-free interest rate $r$ can be estimated as: \begin{equation}\label{call_option}
    C(t) = \Exp^{\mathbb{Q}} \left[ e^{-r(T-t)} \max\{S_T-K,0\} | I_t\right].
\end{equation}
In this study, we are not able to get any closed-form solution for the call option price. Therefore, we estimate the option price through a simulation approach. We first simulate the stock price processes $S_{t}, t \leq T$ using the stochastic differential equation \eqref{StinQ}, and then compute the call option price using \eqref{call_option}.
\begin{remark}
From Girsanov Theorem \ref{Girsanov-theorem}, the call option price at $t=0$,
\begin{align*}
 C(0) & = \int_{\Omega} e^{-rT} \max\{S_T-K,0\}\, d\mathbb{Q}(\omega) \\
& = \int_{\Omega} e^{-rT} \max\{S_T-K,0\}\, M_T(\omega)\, d\mathbb{P}(\omega) \\
& = \int_{\Omega} e^{-rT} \max\{S_T-K,0\}\, \exp\left( -\int_0^T u_s\, dB_s - \frac{1}{2}\int_0^T u^2_s\, ds \right) \, d\mathbb{P}.
\end{align*}
\end{remark}

\section{Empirical evidence}\label{emperical}

To empirically assess the effectiveness of the proposed method, we use daily data of European call options on Nifty 50 index. We select the Nifty 50 index because it is India's most actively traded derivative instrument on the National Stock Exchange (NSE). Our study covers April 11, 2022, to August 30, 2024, focusing on recent historical data. The data is sourced from the publicly available NSE website. The dataset has information about trade date, time to maturity, underlying asset price, strike price, trading volume, total turnover, and the open, high, low, and close prices for all the option contracts for the sample period. In our sample, we restrict the maturity time of options to a maximum of three months. More than three-month maturity option contracts on the NSE are less frequently traded and tend to exhibit lower liquidity.

A significant portion of the option contracts in our sample display low or negligible liquidity. To address this, we use turnover as a proxy for liquidity. Turnover is calculated as the product of contracts traded, lot size, and the underlying price, and is measured in lakhs of Indian rupees. We exclude contracts falling below the seventh decile of the turnover. After applying this filter, the final dataset consists of $62442$ observations, with the minimum turnover value of $28282.52$.  We use the yield on 91-day Treasury bills as the risk-free interest rate. We collect the interest rate data from Reserve Bank of India (RBI) website. The data and the R-codes are available on request.

We group the option data based on moneyness and time to maturity categories. Let \( S_{t} \) be the underlying asset price at time $t$ and \( K \) be the strike price. A call option is classified as at-the-money (ATM) if \( \frac{S_{t}}{K} \in (0.97, 1.03) \), out-of-the-money (OTM) if \( \frac{S_{t}}{K} \leq 0.97 \), and in-the-money (ITM) if \( \frac{S_{t}}{K} \geq 1.03 \) (see: Bakshi \textit{et al.} \cite{bakshi1997empirical}).

Based on time to maturity, option contracts are classified into five categories: A ($0$ day $<$ Time to maturity $\leq 7$ days); B ($7$ days $<$ Time to maturity $\leq 15$ days); C ($15$ days $<$ Time to;  maturity $\leq 30$ days; D ($30$ days $<$ Time to maturity $\leq 60$ days); and E ($60$ days $<$ Time to maturity $\leq 90$ days).

Tables 1 and 2 report the time-to-maturity
distribution and the moneyness distribution of the traded option contracts, respectively.

\begin{table}[H]
\centering
\caption{Time to maturity distribution of the option contracts.}
\label{tab:ttmdistribution}
\begin{tabular}{lccccc}
\toprule
& A & B & C & D & E \\
\midrule
Number of contracts & 32142 & 17507 & 9548 & 2965 & 280 \\
Percentage of contracts (\%) & 51.4750 & 28.0372 & 15.2910 & 4.7484 & 0.4484 \\
\bottomrule
\end{tabular}
\end{table}

\begin{table}[H]
\centering
\caption{Moneyness distribution of the option contracts.}
\label{tab:moneynessdistribution}
\begin{tabular}{lccc}
\toprule
& ATM & ITM & OTM \\
\midrule
Number of contracts & 31617 & 1131 & 29694 \\
Percentage of contracts (\%) & 50.6342 & 1.8113 & 47.5545 \\
\bottomrule
\end{tabular}
\end{table}

As shown in Table~1, the majority of the traded options ($94.8032\%$) have a time to maturity within categories A to C, while category E contains the smallest proportion of contracts.

Table~2 indicates that trading activity is more concentrated in at-the-money (ATM) and out-of-the-money (OTM) call options compared to in-the-money (ITM) call options.

\subsection{Model comparison using the historical approach}

We evaluate our model's performance against the classical Black–Scholes and Heston models. Under the Pearson diffusion framework, the statistical dynamics of log returns $R_{t}$ is governed by \eqref{pearson-diff1}, an SDE having four parameters: $(\theta, a, \mu, \sigma)$. Our first step is to estimate these parameters. In the historical approach, model parameters are estimated from historical log return data and then used to estimate the model option price. For the historical approach, we need to select an appropriate window size. As mentioned in Hull and Basu \cite{hull2016options}, the estimation window for historical volatility in the Black–Scholes model ranges from 30 to 180 trading days in practice. In this study, we select two rolling window lengths of 90 and 180 trading days. In other words, to estimate the model option price on a given day, we first collect the daily closing prices from the preceding 90 or 180 trading days, compute the corresponding log returns, estimate the model parameters, and then use these parameters to calculate the model option price. In the simulation study, call option prices are computed using $200000$ Monte Carlo replications under the proposed PIV method.

For all three models - PIV, Black–Scholes (BS), and Heston (HS) - parameters are estimated via the maximum likelihood method. In the estimation procedure, we set $dt=\frac{1}{252}$, reflecting the use of daily data assuming that the NSE operates 252 trading days per calendar year.

 \subsection{Out-of-sample performance of historical approach}

The out-of-sample dataset consists of $53310$ and $44680$ observations corresponding to the rolling window sizes of 90 and 180 days, respectively. The pricing error of an option is defined as the difference between the model price and the market price. We consider the closing price of a traded option to be the market price. To evaluate model performance, we employ two error metrics: mean absolute error (MAE) and mean squared error (MSE). Tables 3 and 4 display the performance of all models using rolling window sizes of 90 and 180 days, respectively, categorized by moneyness. Tables 5 and 6 show the corresponding results for the same window sizes, for the time to maturity category.

\begin{table}[H]
\centering
\caption{Model Performance under moneyness category by rolling window approach with size 90.}
\label{tab:moneyness_performance1}
\renewcommand{\arraystretch}{1.2}
\setlength{\tabcolsep}{10pt}
\begin{tabular}{lccc}
\hline
 & \textbf{PIV} & \textbf{BS} & \textbf{HS} \\
\hline
\multicolumn{4}{c}{\textbf{MAE}} \\
ITM & \textbf{36.9693} & 54.6638 & 44.7100 \\
OTM & \textbf{11.4605} & 34.4065 & 14.9941 \\
ATM & \textbf{36.2825} & 81.1552 & 40.6405 \\
ALL & \textbf{24.8806} & 59.1877 & 28.9196 \\
\hline
\multicolumn{4}{c}{\textbf{MSE}} \\
ITM & \textbf{3109.4387} & 7791.0597 & 4982.0560 \\
OTM & \textbf{819.6919} & 6542.1246 & 2399.6011 \\
ATM & \textbf{4018.4329} & 15329.0959 & 6137.0253 \\
ALL & \textbf{2531.3807} & 11154.6008 & 4397.8962 \\
\hline
\end{tabular}

\vspace{1em}
\noindent\parbox{0.9\linewidth}{\footnotesize \textit{Note:} This table reports the mean absolute error (MAE) and mean squared error (MSE) for all models. PIV, BS, and HS denote the Pearson diffusion, Black--Scholes, and Heston models, respectively. ITM, OTM, and ATM refer to in-the-money, out-of-the-money, and at-the-money options. ALL represents the entire out-of-sample dataset. Across the entire out-of-sample dataset, the proportions of ITM, OTM, and ATM options are $1.78\%$, $45.98\%$, and $52.24\%$, respectively. Bold values indicate the lowest MAE and MSE in each row.}
\end{table}

\begin{table}[H]
\centering
\caption{Model Performance under moneyness category by rolling window approach with size 180.}
\label{tab:moneyness_performance2}
\renewcommand{\arraystretch}{1.2}
\setlength{\tabcolsep}{10pt}
\begin{tabular}{lccc}
\hline
 & \textbf{PIV} & \textbf{BS} & \textbf{HS} \\
\hline
\multicolumn{4}{c}{\textbf{MAE}} \\
ITM & \textbf{38.6962} & 48.9129 & 47.8571 \\
OTM & \textbf{12.0293} & 29.5032 & 16.0753 \\
ATM & \textbf{38.2951} & 75.2467 & 43.0964 \\
ALL & \textbf{26.4330} & 54.0835 & 30.9744 \\
\hline
\multicolumn{4}{c}{\textbf{MSE}} \\
ITM & \textbf{3367.1278} & 5055.1497 & 5630.3482 \\
OTM & \textbf{854.5360} & 4231.0854 & 2685.0795 \\
ATM & \textbf{3972.1454} & 10881.1451 & 6967.7207 \\
ALL & \textbf{2551.9969} & 7767.2186 & 5007.4213 \\
\hline
\end{tabular}

\vspace{1em}
\noindent\parbox{0.9\linewidth}{\footnotesize \textit{Note:} This table reports the mean absolute error (MAE) and mean squared error (MSE) for all models. PIV, BS, and HS denote the Pearson diffusion, Black--Scholes, and Heston models, respectively. ITM, OTM, and ATM refer to in-the-money, out-of-the-money, and at-the-money options. ALL represents the entire out-of-sample dataset. Across the entire out-of-sample dataset, the proportions of ITM, OTM, and ATM options are $1.87\%$, $45.19\%$, and $52.94\%$, respectively. Bold values indicate the lowest MAE and MSE in each row.}
\end{table}

\begin{table}[H]
\centering
\caption{Model Performance under time-to-maturity category by rolling window approach with size 90.}
\label{tab:ttm_performance1}
\begin{tabular}{lccc}
\hline
 & \textbf{PIV} & \textbf{BS} & \textbf{HS} \\
\hline
\multicolumn{4}{c}{\textbf{MAE}} \\
A & \textbf{11.2749} & 32.0895 & 13.7173 \\
B & \textbf{28.5835} & 72.2835 & 35.5088 \\
C & \textbf{47.5368} & 97.9993 & 52.2150 \\
D & 61.6959 & 131.1378 & \textbf{60.4702} \\
E & \textbf{128.1301} & 158.6207 & 160.7986 \\
\hline
\multicolumn{4}{c}{\textbf{MSE}} \\
A & \textbf{689.8709} & 3466.8727 & 1180.3141 \\
B & \textbf{2464.4858} & 11985.3773 & 4908.5419 \\
C & \textbf{5872.2965} & 23112.1791 & 9431.0848 \\
D & \textbf{8866.9861} & 40282.0234 & 13405.8085 \\
E & \textbf{26497.2694} & 86574.5441 & 55562.8110 \\
\hline
\end{tabular}

\vspace{1em}
\noindent\parbox{0.9\linewidth}
{\footnotesize \textit{Note:} This table presents the mean absolute error (MAE) and mean squared error (MSE) for all models. PIV, BS, and HS represent the Pearson diffusion, Black--Scholes, and Heston models, respectively. Categories A through E classify options based on time to expiry as follows: A (7 days or less), B (more than 7 to 15 days), C (more than 15 to 30 days), D (more than 30 to 60 days), and E (more than 60 to 90 days). Across the entire out-of-sample dataset, the proportions of option contracts corresponding to time-to-maturity categories A, B, C, D, and E are \(50.78\%\), \(28.18\%\), \(15.72\%\), \(4.79\%\), and \(0.52\%\), respectively. Bold values highlight the lowest MAE and MSE for each row.}
\end{table}

\begin{table}[H]
\centering
\caption{Model Performance under time-to maturity category by rolling window approach with size 180.}
\label{tab:ttm_performance2}
\begin{tabular}{lccc}
\hline
 & \textbf{PIV} & \textbf{BS} & \textbf{HS} \\
\hline
\multicolumn{4}{c}{\textbf{MAE}} \\
A & \textbf{12.0051} & 28.1144 & 14.7708 \\
B & \textbf{30.6041} & 65.0758 & 38.5824 \\
C & \textbf{49.5831} & 90.8299 & 54.6517 \\
D & 65.4378 & 127.3506 & \textbf{62.6735} \\
E & \textbf{114.3168} & 148.2005 & 155.4543 \\
\hline
\multicolumn{4}{c}{\textbf{MSE}} \\
A & \textbf{729.9839} & 2429.8010 & 1371.3447 \\
B & \textbf{2468.0535} & 8254.1154 & 5687.6544 \\
C & \textbf{5686.9298} & 15902.2269 & 10627.6116 \\
D & \textbf{9194.8577} & 28308.5065 & 14623.8995 \\
E & \textbf{23047.8746} & 50836.7810 & 53501.0383 \\
\hline
\end{tabular}

\vspace{1em}
\noindent\parbox{0.9\linewidth}
{\footnotesize \textit{Note:} This table presents the mean absolute error (MAE) and mean squared error (MSE) for all models. PIV, BS, and HS represent the Pearson diffusion, Black--Scholes, and Heston models, respectively. Categories A through E classify options based on time to expiry as follows: A (7 days or less), B (more than 7 to 15 days), C (more than 15 to 30 days), D (more than 30 to 60 days), and E (more than 60 to 90 days). Across the entire out-of-sample dataset, the proportions of option contracts corresponding to time-to-maturity categories A, B, C, D, and E are \(50.43\%\), \(28.25\%\), \(15.75\%\), \(5.01\%\), and \(0.57\%\), respectively. Bold values highlight the lowest MAE and MSE for each row.}
\end{table}

Tables 3 and 4 show that the proposed method consistently outperforms the classical Black–Scholes model and the Heston stochastic volatility model across all moneyness categories, for both two error metrics and rolling window sizes. Likewise, Tables 5 and 6 indicate that the PIV method delivers superior performance across all time-to-maturity categories, under both error functions and window sizes, except for Category D (options with maturities between 30 and 60 days) under the MAE criterion.

Next, we apply the Diebold and Mariano \cite{diebold2002comparing} test to evaluate whether the differences in predictive accuracy between the competing models are statistically significant. The test is conducted using both absolute error (AE) and squared error (SE) as loss functions.
The null hypothesis of the test states that the forecasts from the PIV and HS (or BS) models have the same predictive accuracy, while the alternative hypothesis states that the PIV model forecasts are more accurate than the HS (or BS) model. Tables 7 and 8 report the p-values of the Diebold and Mariano test comparing the predictive performance of the PIV model with the HS and BS models, using rolling window sizes of 90 and 180 days, respectively.

Tables 7 and 8 indicate that the null hypothesis is rejected at the $1\%$ significance level across all moneyness categories, for both error functions and all rolling window sizes, providing statistical evidence that the forecasts produced by the proposed model are more accurate than those of the HS and BS models at $1\%$ significance level.

\begin{table}[H]
\centering
\caption{p-values of the Diebold-Mariano Test for rolling window size 90.}
\begin{tabular}{lcc}
\toprule
\multicolumn{3}{c}{\textbf{PIV vs HS}} \\
\midrule
         & \textbf{AE}       & \textbf{SE}       \\
ITM      & 4.30E$^{-10}$     & 6.46E$^{-08}$     \\
OTM      & 2.20E$^{-16}$     & 2.20E$^{-16}$     \\
ATM      & 2.20E$^{-16}$     & 2.20E$^{-16}$     \\
ALL      & 2.20E$^{-16}$     & 2.20E$^{-16}$     \\
\addlinespace
\multicolumn{3}{c}{\textbf{PIV vs BS}} \\
\midrule
         & \textbf{AE}       & \textbf{SE}       \\
ITM      & 2.20E$^{-16}$     & 1.76E$^{-11}$     \\
OTM      & 2.20E$^{-16}$     & 2.20E$^{-16}$     \\
ATM      & 2.20E$^{-16}$     & 2.20E$^{-16}$     \\
ALL      & 2.20E$^{-16}$     & 2.20E$^{-16}$     \\
\bottomrule
\end{tabular}

\vspace{1em}
\begin{minipage}{0.9\textwidth}
\footnotesize
\textit{Note:} This table presents the p-values from the Diebold and Mariano test, used to assess the predictive accuracy of the PIV model relative to the HS (PIV vs HS) and BS (PIV vs BS) models. ITM, OTM, and ATM refer to
in-the-money, out-of-the-money, and at-the-money options. ALL represents the
entire out-of-sample dataset. The AE and SE columns report the p-values of the test based on the Absolute Error and Squared Error loss functions, respectively.
\end{minipage}
\end{table}

\begin{table}[H]
\centering
\caption{p-values of the Diebold-Mariano Test for rolling window size 180.}
\begin{tabular}{lcc}
\toprule
\multicolumn{3}{c}{\textbf{PIV vs HS}} \\
\midrule
         & \textbf{AE}       & \textbf{SE}       \\
ITM      & 6.77E$^{-10}$     & 1.07E$^{-07}$     \\
OTM      & 2.20E$^{-16}$     & 2.20E$^{-16}$     \\
ATM      & 2.20E$^{-16}$     & 2.20E$^{-16}$     \\
ALL      & 2.20E$^{-16}$     & 2.20E$^{-16}$     \\
\addlinespace
\multicolumn{3}{c}{\textbf{PIV vs BS}} \\
\midrule
         & \textbf{AE}       & \textbf{SE}       \\
ITM      & 6.88E$^{-11}$     & 4.68E$^{-06}$     \\
OTM      & 2.20E$^{-16}$     & 2.20E$^{-16}$     \\
ATM      & 2.20E$^{-16}$     & 2.20E$^{-16}$     \\
ALL      & 2.20E$^{-16}$     & 2.20E$^{-16}$     \\
\bottomrule
\end{tabular}

\vspace{1em}
\begin{minipage}{0.9\textwidth}
\footnotesize
\textit{Note:} This table presents the p-values from the Diebold and Mariano test, used to assess the predictive accuracy of the PIV model relative to the HS (PIV vs HS) and BS (PIV vs BS) models. ITM, OTM, and ATM refer to in-the-money, out-of-the-money, and at-the-money options, respectively. ALL represents the entire out-of-sample dataset. The AE and SE columns report the p-values of the test based on the Absolute Error and Squared Error loss functions, respectively.
\end{minipage}
\end{table}

\subsection{Out-of-sample performance using Implied approach}

The historical approach requires selecting a window size to estimate option model parameters, which can have practical challenges due to its stringent requirements on historical data (Bakshi \textit{et al.} \cite{bakshi1997empirical}). To address this limitation, practitioners have traditionally opted to use the implied approach, where model parameters are estimated by inverting the option pricing model using recent market option prices. The implied parameters may be interpreted as the instant market estimate of the future price movement of the underlying asset. This method not only significantly reduces data requirements but has also been shown to improve performance (see: Bates \cite{bates1996jumps, bates1996testing}).

 For the case of the implied approach, we follow Bakshi \textit{et al.} to estimate the model parameters. \cite{bakshi1997empirical}. To estimate the call option price for the day $t$, we first collect all options traded on the preceding day. Let $C_{market,t-1}$  denote the observed market price and $C_{model,t-1}$ be model option price  both at time $t-1$. Let $\Theta$ be the vector of model parameters. The implied parameters of each model for day $t-1$ are estimated by minimizing the sum of squared pricing errors:

 \[
\hat{\Theta} = \underset{\Theta}{\min} \sum \left(C_{\text{market},t-1} - C_{\text{model},t-1}\right)^2
\] where the summation is taken over all options traded on day $t-1$. Next we utilize those estimated implied parameters to estimate the option price of day $t$. Under the implied approach, the out-of-sample dataset comprises a total of $62103$ observations.

Tables 9 and 10 report the out-of-sample forecast performance of all three models across moneyness and time-to-maturity categories, respectively. Table 11 presents the p-values of the Diebold and Mariano test, assessing whether the differences in predictive accuracy among the models are statistically significant.

Table 9 demonstrates that the PIV model outperforms the other models for at-the-money (ATM) options as well as across the entire out-of-sample dataset under both error metrics. In contrast, the HS model delivers better forecasts for in-the-money (ITM) and out-of-the-money (OTM) options. Given that over \(50\%\) of the options in both the full sample (Table 2) and the out-of-sample dataset (Table 9) fall within the ATM category, the superior performance of the proposed model in this segment is particularly significant.
Table 11 presents statistical evidence confirming that the PIV model's forecasts for ATM options are more accurate than those of the HS and BS models at $1\%$ significance level.

For data grouped by time-to-maturity (Table 10), the PIV method outperforms others in categories A and B under the MAE metric, which together account for over $79\%$ of traded options in both the full sample (Table 1) and the out-of-sample (Table 10) dataset. The proposed PIV model delivers superior performance in the groups with the highest trading activity, highlighting its effectiveness within the implied approach framework as well. However, under the MSE metric, the proposed method does not exhibit superior performance, the BS model performs better in categories B and C, while the HS model outperforms in the remaining three categories.

\begin{table}[H]
\centering
\caption{Model Performance under moneyness category by implied approach.}
\label{tab:moneynessperformance3}
\renewcommand{\arraystretch}{1.2}
\setlength{\tabcolsep}{10pt}
\begin{tabular}{lccc}
\hline
 & \textbf{PIV} & \textbf{BS} & \textbf{HS} \\
\hline
\multicolumn{4}{c}{\textbf{MAE}} \\
ITM & 27.5377 & 26.9758 & \textbf{26.1722} \\
OTM & 5.6576 & 5.4501 & \textbf{4.9065} \\
ATM & \textbf{15.0136} & 16.9099 & 15.8467 \\
ALL & \textbf{10.7919} & 11.6430 & 10.8316 \\
\hline
\multicolumn{4}{c}{\textbf{MSE}} \\
ITM & 1628.9130 & 1565.7580 & \textbf{1517.8892} \\
OTM & 189.8867 & 174.3626 & \textbf{131.9044} \\
ATM & \textbf{583.6524} & 634.4850 & 623.1860 \\
ALL & \textbf{415.4041} & 432.6046 & 405.8215 \\
\hline
\end{tabular}

\vspace{1em}
\noindent\parbox{0.9\linewidth}{\footnotesize \textit{Note:} This table reports the mean absolute error (MAE) and mean squared error (MSE) for all models. To compute the option price on a given day for all models, we first estimate the implied parameters using all options traded on the previous day, and then use these parameters to calculate the option price. PIV, BS, and HS denote the Pearson diffusion, Black--Scholes, and Heston models, respectively. ITM, OTM, and ATM refer to in-the-money, out-of-the-money, and at-the-money options. ALL represents the entire out-of-sample dataset. Across the entire out-of-sample dataset, the proportions of ITM, OTM, and ATM options are \(1.82\%\), \(47.56\%\), and \(50.62\%\), respectively. Bold values indicate the lowest MAE and MSE in each row.}
\end{table}

\begin{table}[H]
\centering
\caption{Model Performance under time-to-maturity category by implied approach.}
\label{tab:ttm_performance3}
\begin{tabular}{lccc}
\hline
 & \textbf{PIV} & \textbf{BS} & \textbf{HS} \\
\hline
\multicolumn{4}{c}{\textbf{MAE}} \\
A & \textbf{6.6631} & 8.3303 & 6.9006 \\
B & \textbf{12.1654} & 12.3763 & 13.4534 \\
C & 17.1369 & 16.8262 & \textbf{16.0479} \\
D & 23.0805 & 22.4659 & \textbf{18.8199} \\
E & 50.3616 & 53.0035 & \textbf{34.0092} \\
\hline
\multicolumn{4}{c}{\textbf{MSE}} \\
A & 204.2751 & 246.4082 & \textbf{196.3507} \\
B & 420.3412 & \textbf{404.1864} & 480.5811 \\
C & 681.9721 & \textbf{675.7133} & 735.1835 \\
D & 1248.3535 & 1257.3568 & \textbf{936.6570} \\
E & 6308.4326 & 6443.3622 & \textbf{2822.3595} \\
\hline
\end{tabular}

\vspace{1em}
\noindent\parbox{0.9\linewidth}
{\footnotesize \textit{Note:} This table presents the mean absolute error (MAE) and mean squared error (MSE) for all models by implied approach. To compute the option price on a given day for all models, we first estimate the implied parameters using all options traded on the previous day, and then use these parameters to calculate the option price. PIV, BS, and HS represent the Pearson diffusion, Black--Scholes, and Heston models, respectively. Categories A through E classify options based on time to expiry as follows: A (7 days or less), B (more than 7 to 15 days), C (more than 15 to 30 days), D (more than 30 to 60 days), and E (more than 60 to 90 days). Across the entire out-of-sample dataset, the proportions of option contracts corresponding to time-to-maturity categories A, B, C, D, and E are \(51.40\%\), \(28.06\%\), \(15.33\%\), \(4.76\%\), and \(0.45\%\), respectively. Bold values highlight the lowest MAE and MSE for each row.}
\end{table}

\begin{table}[H]
\centering
\caption{p-values of the Diebold-Mariano Test for implied approach.}
\label{tab:dm_test_pvalues}
\begin{tabular}{lcc}
\toprule
\multicolumn{3}{c}{\textbf{PIV vs HS}} \\
\midrule
         & \textbf{AE}       & \textbf{SE}       \\
ITM      & 0.9993            & 0.9607            \\
OTM      & 1.0000            & 1.0000            \\
ATM      & 2.20E$^{-16}$     & 0.0002            \\
ALL      & 0.2003            & 0.8884            \\
\addlinespace
\multicolumn{3}{c}{\textbf{PIV vs BS}} \\
\midrule
         & \textbf{AE}       & \textbf{SE}       \\
ITM      & 0.9988            & 0.9990            \\
OTM      & 1.0000            & 1.0000            \\
ATM      & 2.20E$^{-16}$     & 2.20E$^{-16}$     \\
ALL      & 2.20E$^{-16}$     & 3.88E$^{-14}$     \\
\bottomrule
\end{tabular}

\vspace{1em}
\begin{minipage}{0.9\textwidth}
\footnotesize
\textit{Note:} This table presents the p-values from the Diebold and Mariano test, used to assess the predictive accuracy of the PIV model relative to the HS (PIV vs HS) and BS (PIV vs BS) models. ITM, OTM, and ATM refer to in-the-money, out-of-the-money, and at-the-money options, respectively. ALL represents the entire out-of-sample dataset. The AE and SE columns report the p-values of the test based on the Absolute Error and Squared Error loss functions, respectively.
\end{minipage}
\end{table}

\section{Conclusion}\label{conclusion}

We propose a novel method for estimating call option prices by modeling the log return dynamics using Pearson diffusion processes, characterized by a linear drift and a quadratic squared diffusion coefficient. The advantage of this class of diffusion is its ability to capture the excess kurtosis and skewness observed in the return distribution. We also verify the Novikov condition, a key requirement for risk-neutral pricing, thereby ensuring that the model does not admit arbitrage opportunities. Further, we establish the existence of a unique strong solution for the stock price process under the risk-neutral measure. We apply our method to Nifty 50 index options and compare the results to the classical Black-Scholes and Heston stochastic volatility models using historical and implied approaches.

Under the historical approach, the proposed method consistently outperforms competing models across all rolling window sizes and evaluation metrics. These results hold across option categories defined by moneyness and time-to-maturity (except 30 to 60 days maturity options).

Under the implied approach, the PIV method demonstrates superior performance for at-the-money options, which account for over $50\%$ of the observations, as well as across the entire out-of-sample dataset. For the time-to-maturity category, the PIV method outperforms other methods for short-term maturity options (one-week and two-weeks maturity options), which together constitute more than $79\%$ of the dataset.  The Heston model exhibits better performance for long-term maturity options.

The findings based on historical and implied approaches show the practical relevance of our approach and suggest that it offers a more effective alternative for option pricing. \\

\textbf{Acknowledgment:} We gratefully acknowledge Professor Malay Bhattacharyya for introducing us to the Pearson diffusion processes. Barun Sarkar would like to acknowledge the seed grant of Indian Institute of Technology Madras. Suprio Bhar acknowledges the support of the Matrics grant MTR/2021/000517 from the Science and Engineering Research Board (Department of Science \& Technology, Government of India).

\appendix
\section{Appendix}\label{App}

\begin{proof}[Proof of Theorem \ref{localmart-discountStck}]
Girsanov's Theorem \ref{Girsanov-theorem} proposed a new Brownian motion $\widehat B_t, t \leq T$ under $\Q$ of the form
\[d\widehat{B_t}= u_t\, dt + dB_t,\, t \leq T.\]
From \eqref{eq:discounted_process_drift_zero},
\begin{align*}
d\left( e^{-rt} S_t\right)
& = -re^{-rt} S_t\, dt + e^{-rt} S_t \left\{ -\theta(R_t-\mu) +  \sigma^2\theta a \left(1+R_t^2 \right)\right\}\, dt  \\
& \quad + e^{-rt} S_t \sigma \sqrt{2\theta a(1+R_t^2)}\, \left(d\widehat{B_t} - u_t \, dt\right),
\end{align*}
and from \eqref{expression-for-ut}
\[u_t = \frac{-r  -\theta(R_t-\mu) +  \sigma^2\theta a \left(1+R_t^2 \right)}{ \sigma \sqrt{2\theta a(1+R_t^2)}}.
\]
The proposed probability measure $\Q$ is given by
\begin{equation}\label{risk-neutral-measure}
  d\mathbb{Q}=M_T\, d\mathbb{P},
\end{equation}
where
\begin{equation}\label{martingale-in-risk-neutral}
  M_T:= \exp\left( -\int_0^T u_s\, dB_s - \frac{1}{2}\int_0^T u^2_s\, ds \right).
\end{equation}
If we can show that $M_t,\, t \leq T$ is a $\mathbb{P}$-martingale, then $\mathbb{Q}$ becomes a probability measure. Under $\Q$,
\begin{equation}\label{newYt}
d\left( e^{-rt} S_t\right) =  e^{-rt} S_t \sigma \sqrt{2\theta a(1+R_t^2)}\, d\widehat{B_t}
\end{equation}
is a local martingale. This achieves our target provided we show that $M_t,\, t \leq T$ is a $\mathbb{P}$-martingale. Observe that from \eqref{expression-for-ut}
 \begin{align}\label{condndec24}
\begin{split}
u_t^2 &=  \left|\frac{-r  -\theta(R_t-\mu) +  \sigma^2\theta a \left(1+R_t^2 \right)}{ \sigma \sqrt{2\theta a(1+R_t^2)}} \right|^2 \\
& = \left|\frac{-r +\mu\theta -\theta R_t +  \sigma^2\theta a \left(1+R_t^2 \right)}{ \sigma \sqrt{2\theta a(1+R_t^2)}}\right|^2 \\
& \leq 3\left\{ \left| \frac{\mu\theta-r}{\sigma \sqrt{2\theta a(1+R_t^2)}}\right|^2 + \left|\frac{\theta R_t}{\sigma \sqrt{2\theta a(1+R_t^2)}} \right|^2 + \left|\sigma \sqrt{\frac{a\theta}{2}}\, \sqrt{1+R_t^2} \right|^2\right\} \\
& \leq 3\left\{\frac{(\mu\theta-r)^2}{2\sigma^2\theta a} + \frac{\theta^2}{2\sigma^2\theta a} + \frac{\sigma^2a\theta}{2} \left(1+R_t^2\right)\right\} \\
& \leq K(\mu,\theta,r,\sigma,a)\,   \left(1+R_t^2\right),
\end{split}
 \end{align}
where $K(\mu,\theta,r,\sigma,a) > 0$ is a constant depending only on $\mu,\theta,r,\sigma$ and $a$. Therefore for Novikov’s condition, using \eqref{condndec24},
\begin{align*}
&\Exp \left[\exp\left( \frac{1}{2} \int_0^T u^2_t\, dt   \right)\right]\\
&= \Exp \left[ \exp \left( \frac{1}{2}\int_0^T \left|
 \frac{ -r  -\theta(R_t-\mu) +  \sigma^2\theta a \left(1+R_t^2 \right)}{\sigma \sqrt{2\theta a(1+R_t^2)}}\right|^2 dt\right)\, \right] \\
 & \leq  \Exp \left[ \exp \left(\frac{1}{2} K(\mu,\theta,r,\sigma,a)\int_0^T  \left(1+R_t^2\right)  dt\right)\right].
\end{align*}
From \eqref{Rt-integrability} and \eqref{condndec24}, we know that
\[\pp\left( \int_0^T  \left(1+R_t^2\right) dt<\infty\right) =1.\]
Therefore,
\begin{equation}\label{integrabilityX}
\pp\left( \int_0^T  u_t^2 dt<\infty\right) =1.
\end{equation}
Following  \cite[Theorem 2.1]{MR2446690}, define
\[ \chi_N(t):=
\begin{cases}
 1,\ \ \text{if}\ \int_0^t  u_s^2 ds\leq N \\
 0,\ \text{otherwise}.
\end{cases}
     \]
Therefore, from \eqref{integrabilityX}
\[\lim_{N\to\infty}\pp\left( \int_0^T  u_t^2 dt\leq N\right) =1.\]
Then, from definition of $\chi_N$
\[\int_0^T\left( \chi_N(t)u_t\right)^2\, dt\leq N,\ \ \text{$P$-a.s.}\]
Now, define
\[M^N_t:= \exp\left( -\int_0^t \chi_N(s)u_s\, dB_s - \frac{1}{2}\int_0^t \left(\chi_N(s)u_s\right)^2\, ds \right).\]
Here, the Novikov condition is satisfied
\[\Exp \left[\exp\left(\frac{1}{2}\int_0^T \left(\chi_N(s)u_s\right)^2\, ds\right) \right]< \infty,\]
and consequently, $\{M^N_t\}_{t\in[0,T]}$, is an $\mathcal{L}^2$ exponential martingale \cite[Chapter III., Theorem 5.3]{MR1011252}. Here, $\mathcal{L}^2$ denotes the space of square integrable random variables. Define
\[d\mathbb{Q}^N:=M_T^N\, d\mathbb{P}\]
and note that $\Exp \left(M^N_T\right)=1$. Hence,
\begin{align*}
1 & = \Exp \left(M^N_T\right) \\
& =  \Exp \left(\chi_N(T)M^N_T\right) + \Exp \left((1-\chi_N(T))M^N_T\right) \\
& = \Exp \left(\chi_N(T)M_T\right) + \pp\left( \chi_N(T) = 0\right)
\end{align*}
Now, by the Monotone Convergence Theorem, $\lim_{N\to\infty}\Exp \left(\chi_N(T)M_T\right) = \Exp \left(M_T\right)$ and by
\eqref{integrabilityX}, \[\lim_{N\to\infty} \pp\left( \chi_N(T) = 0\right) = \lim_{N\to\infty} \pp\left( \int_0^T  u_t^2 dt > N\right) = 0.\]
Therefore, we conclude
\[ \lim_{N\to\infty}  \Exp \left(M^N_T\right) =  \Exp \left(M_T\right) = 1.\]
Thus, $M_t,\, t \leq T$ is a $\mathbb{P}$-martingale.
\end{proof}

\bibliographystyle{plain}

\bibliography{references}

\end{document}